\documentclass{llncs}
\usepackage{sidecap}
\usepackage{wrapfig}

\usepackage{amssymb,amsmath}
\usepackage{array}  
\usepackage{tikz} 
\usetikzlibrary{snakes,shapes,backgrounds,arrows,matrix,petri} 
\usepackage{extarrows}
\usepackage{xargs} 
\usepackage{xifthen}

\newcommand{\figref}[1]{Figure \ref{#1}}

\newcommand{\DEF}[1]{{\bfseries{\textcolor{black}{#1}}}}

\def\ie{{\it i.e.}~}  
\def\cf{{\it cf.}~}
\def\eg{{\it e.g.}~}  
\newcommand{\scc}[0]{{\sc scc}}
\def\free{F}
\def\grows{G}
\def\destroys{D}

\def\merges{M}

\def\freeS{\free-sensitivity}
\def\growsS{\grows-sensitivity}
\def\destroysS{\destroys-sensitivity}
\def\mergesS{\merges-sensitivity}

\def\sy{synchronous}
\def\as{asynchronous}
\def\seqable{sequentialisable}

\def\normal{normal} 
\def\frableC{critical cycle} 
\def\fredC{\frableC} 
\newcommand{\Ban}[0]{{\sc ban}}

\newcommand{\tg}[0]{transition graph}

\newcommand{\sig}[0]{{\sc atg}}
\newcommand{\gig}[0]{{\sc etg}}
\newcommand{\ltf}[0]{transition function}

\newcommand{\fullstop}{\text{. }} 
\newcommand{\comma}{\text{, }} 
\newcommand{\set}[1]{\ensuremath{\{\, #1\, \}}}

\renewcommand{\bar}[2]{\ensuremath{\overline{#1}^{\, \scriptstyle #2}}} 

\renewcommand{\implies}[0]{\Rightarrow}

\def\Rar{\ensuremath{ \ \Rightarrow\ }}

\def\sbool{\ensuremath{\mathbf{s}}}

\DeclareMathOperator{\Hdist}{\mathrm{d}}
\DeclareMathOperator{\Hdiff}{\mathrm{D}}
\def\B{\ensuremath{\mathbb{B}}}
\def\Bn{\ensuremath{\B^n}}

\def\naturals{\ensuremath{\mathbb{N}}}
\newcommand{\att}[1]{\ensuremath{[#1]}}
\newcommand{\atta}[1]{\ensuremath{[#1]_{\mathrm{a}}}}
\def\Aa{\ensuremath{\mathcal{A}_{\mathrm{a}}}}
\def\Ae{\ensuremath{\mathcal{A}}}
\def\Ba{\ensuremath{\mathcal{B}_{\mathrm{a}}}}
\def\Be{\ensuremath{\mathcal{B}}}
\def\orbita{\ensuremath{\mathcal{O}_{\mathrm{a}}}}
\def\orbit{\ensuremath{\mathcal{O}}}
\def\N{\ensuremath{\mathcal{N}}}
\def\G{\ensuremath{ \mathbf{G}}}   
\def\H{\ensuremath{ \mathbf{H}}}   
\def\V{\ensuremath{ \mathbf{V}}}   
\def\C{\ensuremath{ \mathbf{C}}}   
 \def\W{\ensuremath{ \mathrm{W}}}   
\def\A{\ensuremath{ \mathbf{A}}}   
\def\VA{\ensuremath{ (\V,\A)}}
\newcommand{\Vin}[0]{\ensuremath{\V^{-}}}
\newcommand{\sign}[0]{\ensuremath{\mathrm{sign}}}
\def\U{\ensuremath{ \mathbf{U}}} 
\def\nU{\ensuremath{\overline{\U}}}
\def\u{\ensuremath{\mathrm{u}}}
\def\frus{\mathbf{\scriptstyle FRUS}}

\newcommand{\uf}[0]{\ensuremath{  \mathrm{f}\hspace{0.1ex}}}
\newcommand{\uh}[0]{\ensuremath{  \mathrm{h}\hspace{0.1ex}}}
\def\TG{\ensuremath{\mathcal{T}}}    
\def\T{\ensuremath{\mathrm{T}}}      
\newcommand{\AT}[0]{\ensuremath{\T_{\mathrm{a}}}}   

\newcommand{\SIG}[0]{\ensuremath{ \TG_{\mathrm{a}}} }
\newcommand{\SIGbis}[0]{\ensuremath{ \TG_{\mathrm{a}}\,'} }
\newcommandx{\trans}[2][1=,2=0.5]{\,
  \raisebox{.4ex}{
    \begin{tikzpicture}
      \path (0,0) edge[->, >=angle 60] node[above,pos=0.5,font=\scriptsize]
      {\ensuremath{\scriptstyle #1}} (#2,0);
    \end{tikzpicture}}\ }

\newcommandx{\transRT}[2][1=,2=white]{\,
  \raisebox{.4ex}{
    \begin{tikzpicture}[descr/.style={fill=#2,inner sep=2.5pt}]
      \path (0,0) edge[->>, >=angle 60] node[above,font=\scriptsize]
           {\ensuremath{\scriptstyle #1}} (0.6,0);
  \end{tikzpicture}}\ }

\newcommandx{\seq}[2][1=,2=0.5]{\,
  \raisebox{.4ex}{
    \begin{tikzpicture}
      \path (0,0) edge[-open triangle 60] node[above,pos=0.45,font=\scriptsize]
      {\ensuremath{\scriptstyle #1}} (#2,0);
    \end{tikzpicture}}\ }

\newcommandx{\seqRT}[2][1=,2=0.6]{\,
  \raisebox{.4ex}{
    \begin{tikzpicture}
     \path  (0,0) +(-#2,0) edge[-open triangle 60] node[above,pos=0.45,font=\scriptsize]
      {\ensuremath{\scriptstyle #1}} (-0.1,0) ;\draw[-open triangle 60] (-0.1,0) -- +(.2,0);
    \end{tikzpicture}}\ }

\newcommandx{\pll}[2][1=,2=0.5]{\,
  \raisebox{.4ex}{
    \begin{tikzpicture}
      \path (0,0) edge[-triangle 60] node[above,pos=0.4,font=\scriptsize]
      {\ensuremath{\scriptstyle #1}} (#2,0);
    \end{tikzpicture}}\ }

\newcommandx{\pllRT}[2][1=,2=0.6]{\,
  \raisebox{.4ex}{
    \begin{tikzpicture}
     \path  (0,0) +(-#2,0) edge[-triangle 60] node[above,pos=0.45,font=\scriptsize]
      {\ensuremath{\scriptstyle #1}} (-0.1,0) ;\draw[-triangle 60] (-0.1,0) -- +(.2,0);
    \end{tikzpicture}}\ }
\newcommandx{\nonseq}[2][1=,2=0.5]{
   \raisebox{.1ex}{
    \begin{tikzpicture}
      \path (0,0) +(-#2,0) edge[-,line width=1mm] node[above,pos=0.52,font=\scriptsize]
      {\ensuremath{\scriptstyle #1}} (-0.1,0);
      \draw[-triangle 60] (-0.1,0) -- +(.2,0);
    \end{tikzpicture}
 }}

\begin{document}
\bibliographystyle{splncs_srt}
\setlength{\parindent}{0cm}

%
%
%
\title{Synchronism {\it vs} Asynchronism\\ in Boolean automata networks}
\titlerunning{Synchronism {\it vs} Asynchronism}
\author{Mathilde Noual\inst{1,2}}
\authorrunning{M. Noual}
\institute{Laboratoire I3S, UMR 7271 - UNS CNRS, Universit\'e de Nice-Sophia
  Antipolis, 06900 Sophia Antipolis, France; \texttt{noual@i3s.unice.fr} \and IXXI, Institut
  rh\^one-alpin des syst\`emes complexes, Lyon, France}
\maketitle
\begin{abstract} We show that synchronism can 
  significantly impact on network behaviours, in particular by filtering
  unstable attractors induced by a constraint of asynchronism. We investigate
  and classify the different possible impacts that an addition of synchronism
  may have on the behaviour of a Boolean automata network. We show how these
  relate to some strong specific structural properties, thus supporting the idea
  that for most networks, synchronism only shortcuts \as{} trajectories. We end
  with a discussion on the close relation that apparently exists between
  sensitivity to synchronism and non-monotony.
\end{abstract}
%
%

\subsection*{Introduction}

In works involving automata networks, synchronism has often either been
considered as a founding hypothesis, as in \cite{Kauffman1969} and the many
studies that followed in its lead, or, on the contrary, in lines with
\cite{Thomas1983}, it has been disregarded altogether to the benefit of pure
asynchrony. In some applied contexts, \emph{theoretical synchronism}
is sometimes understood as \emph{simultaneity} although this restrictive
interpretation relies on a formalisation of duration which conflicts with the
discrete nature of automata networks.  More simply, the synchronous occurrence
of two changes in a network can be regarded as occurrences that are close enough
in time to disallow any other significant event in between them. This naturally
defines a much more general notion of time flow that is strongly \emph{relative}
to the set of events underwent by the network. And thus it yields substantial
representational capacity to the notion of synchronism, justifying the attention
that we propose to give to it in this paper.\medskip

Comparisons have been made between different kinds of ways of updating automata
states, involving variable degrees of synchronism in both probabilistic
\cite{Chassaing2007,Fates2009,Fates2006,Regnault2006,Schabanel2009} (with
cellular automata) and deterministic frameworks
\cite{Aracena2009,Elena2004M,Elena2009T,Goles2008,Robert1986B,Robert1995B}.  In
particular, for the algorithmic purpose of finding the shortest path to a stable
configuration, Robert\cite{Robert1995B} compared Boolean automata network
behaviours under the parallel and sequential update schedules.  In this context,
he noted three ``frequent (but not systematic) phenomena'' that can be observed
through the effect of parallelisation: the ``\emph{bursting}'', the
``\emph{aggregation}'' and the ``\emph{implosion}'' of attraction basins. Here,
we focus on attractors, both stable and unstable.  And considering more
generally state transition systems rather than just deterministic dynamical
systems, we propose to investigate synchronism {\it per se}, and analyse its
\emph{input} to the design of Boolean automata network behaviours. More
precisely, we propose to consider \emph{asynchronous} transition graphs
representing the set of all punctual and atomic events of a network and we
propose to explore the consequences of adding to it a synchronous transition,
representing a new possibility to perform a punctual but non-atomic change.  We
propose to identify the cases where such an addition of synchronism changes
substantially a network's possible asymptotic behaviours or its evolutions
towards them. Thus, we are looking for networks for which synchronism does not
just shortcut \as{} trajectories but rather also adds some new ones that can not
be \emph{mimicked} asynchronously. This leads to classify the possible impacts
of \emph{non-sequentialisable transitions} and then the sensitivity of networks
to synchronism.






\subsection*{Preliminaries}


\paragraph*{Notations --} 
By default, $\V=\set{0,\ldots,n-1}$ denotes a set of $n\in\naturals$ automata
numbered from $0$ to $n-1$. We let $\B=\set{0,1}$. Any $x\in\Bn$ is called a
\DEF{configuration} and its component $x_i\in\B$ is regarded as the \DEF{state}
of automaton $i\in\V$.  In this paper, special attention is paid to switches of
automata states starting in a given configuration. For this reason, we introduce
the following notations:
\begin{multline*}
  \forall x = x_0 \ldots x_{n-1} \in \Bn,  \forall i\in \V,\ \bar{x}{i} = x_0
  \ldots x_{i-1}\, \neg x_i\, x_{i+1} \ldots x_{n-1}\\\text{and } \forall \W =
  \W' \uplus \set{i} \subseteq \V,\ \bar{x}{\W} = \bar{(\bar{x}{i})}{\W'} =
  \bar{(\bar{x}{\W'})}{i} \fullstop
  \label{barre}
\end{multline*}
%
%
Also, to compare two configurations $x,y\in\Bn$, we use: $ \Hdiff(x,y)=
\set{i\in \V;\ x_i\neq y_i}$ and the Hamming distance
$\Hdist(x,y)=|\Hdiff(x,y)|$.
%
%
Finally, to switch values from $\B$ to $\set{-1,1}$, we let
$\sbool:b\in\B\mapsto b-\neg b\in\set{-1,1}$.



\paragraph*{Networks --}
A \DEF{Boolean automata network} (\Ban) of size $n$ is comprised of $n$
interacting automata. Formally, it is a set $\N=\set{\uf_i\,;\, i\in\V}$ of $n$
Boolean functions specifying ``how the net of automata works''. Function
$\uf_i:\B^n\to \B$ specifies the behaviour of automaton $i\in\V$ in any
configuration $x\in\Bn$. It is called the \DEF{\ltf{}} of $i$.  We 
focus on functions that are \emph{locally monotone} w.r.t.  all their
components, \ie $\forall i\in\V, \forall j\in\V$ we assume: 
\begin{eqnarray}
  \text{that either} && \forall x\in\Bn,\ \sbool(x_j) \cdot
  (\uf_i(x)-\uf_i(\bar{x}{j})) \geq 0 \label{Eq: pos} \\
  \text{or} && \forall x\in\Bn,\ \sbool(x_j) \cdot
  (\uf_i(x)-\uf_i(\bar{x}{j})) \leq 0
  \fullstop\label{Eq: neg}
  \label{HYP LM}
\end{eqnarray}
At the end of this paper, non-monotony is discussed. Until then, we assume all
\Ban{s} to be \DEF{monotone}, that is, to involve \ltf{s} that are locally
monotone w.r.t. all their components.


\paragraph*{Network structures --} 
The \DEF{structure} of $\N$ is the digraph $\G=\VA$ whose node set is $\V$
(thus, automata are also called nodes) and whose arc set is: $ \A\ =\ \set{
  (j,i)\in\V^2;\, \exists x\in\Bn,\ \uf_i(x)\neq \uf_i(\bar{x}{j})}$.  $\Vin(i)$
denotes the in-neighbourhood of $i\in \V$ in $\G$. The local monotony of \ltf{s}
allows us to \DEF{sign} the arcs of $\G$. $\forall (j,i)\in\A$ and $\forall
x\in\Bn$ s.t. $ \uf_i(x)\neq \uf_i(\bar{x}{j})$, we let
  $ \sign(j,i)= \sbool(x_j)\cdot
  (\uf_i(x)-\uf_i(\bar{x}{j})) = \sbool(x_j)\cdot\sbool(\uf_i(x))$
which equals $+1$ if \eqref{Eq: pos} is satisfied and $-1$ if \eqref{Eq: neg}
is. We let $\sign(j,i)=0$ when $(j,i)\notin\A$. Naturally, we define the sign of
a path or cycle in $\G$ as the product of the signs of the arcs it
involves. Thus, a positive path globally transmits ``information'' directly
whereas a negative one transmits its negation.


\paragraph*{Instabilities and frustrations --}
For every $x\in\Bn$, we define the set: $ \U(x) = $ $ \set{i\in\V;\ \uf_i(x)\neq
  x_i}\fullstop $ Automata in $\U(x)$ are said to be \DEF{unstable} (or
``calling for a change or updating''\cite{Remy2003}) in $x$ and those in $\nU(x)
= \V\setminus \U(x)$ are said to be \DEF{stable} in $x$. Informally, the number
$\u(x)=|\U(x)|$ of instabilities in $x$ can be understood as the \emph{velocity}
or \emph{momentum} of $\N$ in $x$. Configurations $x$ such that $\u(x)=0$ are
called \DEF{stable}.  Our first lemma which will be very useful in the sequel,
relates instabilities to \Ban{} structures. Its proof is simple so we skip it.
\begin{lemma}[loops]
$  
  \forall i\in\V,\, \forall x\in\Bn, ~ i\in\nU(x)\cap\nU(\bar{x}{i})\Rightarrow
  \sign(i,i)=+1 \text{~ and }~i\in\U(x)\cap\U(\bar{x}{i}) \Rar
  \sign(i,i)=-1\fullstop
  $
\label{lem:loop}
\end{lemma}



$\forall x\in\Bn$, we introduce the set of \DEF{arcs that are
  frustrated\cite{Combe1997,Goles1980T,Toulouse1977b,Toulouse1977a}\label{frustrated}
  in $x$}:
\begin{equation}
 \frus(x)\ =\ \set{ (j,i)\in\A\,;\, \sbool(x_j)\cdot \sbool(x_i) = - \sign(j,i)}
  \label{Eq: def frus}
\end{equation} 
Our second preliminary lemma states that adding frustrated arcs incoming an
unstable automaton cannot make it stable. Again, we skip its proof which mainly
relies on the local monotony of \ltf{s}.
\begin{lemma}[frustrations \& instabilities]  $\forall i\in \V,\ \forall x,y\in \Bn$:
$\big(\, i\in\U(x)\ \wedge\ \frus(x)\cap \Vin(i) \subset \frus(y)\cap \Vin(i)\,
  \big)\ \Rar\ i\in\U(y)\fullstop$
\label{lemma: adding frus}
\end{lemma}
%
%


\paragraph*{Transitions and and \tg{s} --}
An \DEF{elementary transition} of a \Ban{} $\N$ represents an \emph{effective},
\emph{punctual} and possible change in $\N$. It is any couple of configurations
$(x,y)\in\Bn\times \Bn$, noted $x\trans y$, which satisfies: $\emptyset \neq
\Hdiff(x,y) \subseteq \U(x)$.  The \DEF{size of an elementary transition}
$x\trans y$ equals $\Hdist(x,y)$.  Digraph $(\Bn,\set{x\trans y;\, x,y\in\Bn})$
is called the \DEF{elementary transition graph} (\gig) of $\N$. It represents
all punctual/elementary events that $\N$ can undergo.  There are two main types
of elementary transitions $x\trans y$. Those of size $\Hdist(x,y)>1$ are called
\DEF{non-atomic} or \DEF{synchronous} and are written $x\pll y$. Those of size
$\Hdist(x,y)=1$ are called \DEF{asynchronous} or \DEF{atomic} and are written
$x\seq y$ (they are s.t. $\exists i\in \V,\,y=\bar{x}{i}$). Digraph
$\SIG=(\Bn,\set{x\seq y;\, x,y\in\Bn})$ is called the \DEF{asynchronous
  transition graph} or \sig. It represents only those events that $\N$ can
undergo which involve only one local automaton state change.  The transitive
closure of $\trans$ (resp. $\seq$) is denoted by $\transRT$ (resp. $\seqRT$).
\DEF{Derivations} are ordered lists of these non necessarily elementary
transitions written $x^0\transRT x^1 \transRT \ldots \transRT x^{\ell-1}
\transRT x^\ell$ but in the sequel, we abuse language and also speak of a
\emph{derivation} $x\transRT y$.








\subsection*{Non-\seqable{} transitions and \frableC{s}}

A \DEF{cycle} of a \Ban{} $\N$ is a sub-graph of its structure $\G$ corresponding
to a closed directed walk, with possibly repeated nodes but no repeated edges.
$\forall x\in\Bn$, we say that a cycle $\C=(\V_{\C},\A_\C)$ of $\N$ is
\DEF{$x$-critical} if: $
\V_\C\subset\U(x)\ \wedge\ \A_\C\subset\frus(x)\fullstop$ Note that for an
isolated cycle, since $|\Vin(i)|=1,\, \forall i\in\V_\C$, a node is unstable if
and only if its sole incoming arc is frustrated.  A \DEF{critical cycle} of $\N$
is one that is $x$-critical for some $x\in\Bn$.  All main results of this paper
mention these types of cycles. This yields some importance to  Proposition
\ref{prop:critical cycles} below which derives from the following which, by
\eqref{Eq: def frus}, holds for any $x$-critical cycle $\C=(\V_\C,\A_\C)$ of
length $\ell$ and sign $\sbool$: $\displaystyle\prod_{(j,i)\in \A_\C} -
\sign(j,i) =$ $(-1)^\ell\times \sbool=\displaystyle\prod_{(j,i)\in
  \A_\C}\sbool(x_j)\cdot \sbool(x_i) = 1\fullstop $
\begin{wrapfigure}[22]{r}{0.3\textwidth}
  \centerline{\scalebox{0.7}{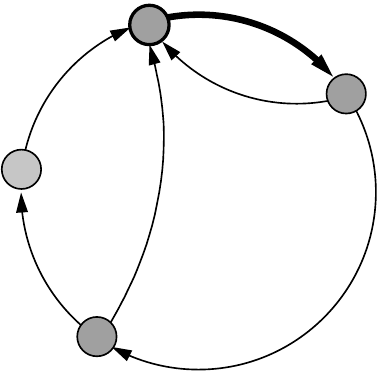}}
  \caption{Signed \Ban{} structure whose Hamiltonian cycle $\C=(\V_{\C},\A_\C)$
    is as in Proposition~\ref{prop:critical cycles}. If $\uf_2:x\mapsto
    x_{3}\wedge(x_0\vee x_1)$, then $\C$ cannot be critical because $2\in \U(x)$
    and $\A_{\C}\subset \frus(x)$ cannot be satisfied at once.}
\label{Ex:Nefrable}
\end{wrapfigure}

\begin{proposition}
  A cycle that is critical is either positive with an even length or negative
  with an odd length.
\label{prop:critical cycles}
\end{proposition}


Let us emphasise that although a positive (resp. negative) cycle with an even
(resp. odd) length is critical when it is isolated, when embedded in larger
structures, it may loose this property (\cf \figref{Ex:Nefrable}).\smallskip




Now, the next result sets the backbone of the article: it shows how critical
cycles are the main structural aspects of a \Ban{} underlying its possibility to
perform synchronous changes that cannot be mimicked asynchronously. First, let
us say that $x\trans y$ is \DEF{\seqable{}} if it is \as{} or if it can be
broken into an derivation $x\transRT y$ involving smaller transitions.  A
synchronous transition $x\pll y$ which is not \seqable{} is called a
\DEF{\normal{}} \DEF{transition} and is rather written $x\nonseq y$. 



\smallskip

\begin{proposition}[\seqable{} transitions and \fredC{s}] Let  $x\pll y$ be a
  \sy{} transition of an arbitrary \Ban{} $\N$. There is a derivation $ x\trans
  \bar{x}{\Hdiff_0} \trans \bar{x}{\Hdiff_0\uplus \Hdiff_1} \trans \ldots \trans
  \bar{x}{\Hdiff_0\uplus \Hdiff_1 \ldots\uplus \Hdiff_{m-1}}=y $ of $\N$ such
  that $\Hdiff(x,y)=\biguplus_{t<m}\Hdiff_t$ and $\forall t<m,\, |\Hdiff_t|>1$
  holds only if all automata of $\Hdiff_t$ belong to the same $x$-\fredC{}.
  \label{prop:romp}
\end{proposition}

A crucial consequence of Proposition~\ref{prop:romp} is that any \sy{}
transition $x\pll y$ is \seqable{} as long as the automata in $\Hdiff(x,y)$ do
not all belong the same $x$-\fredC{}. And it is \emph{totally} \seqable{}
($x\seqRT y$) if no subset of $\Hdiff(x,y)$ induces a $x$-\fredC{}.  Generally,
Proposition~\ref{prop:romp} implies that for a \Ban{} with no \frableC{s} of
size $m\in\naturals$ or less, any synchronous change of $m'\leq m$ automata
states can be totally sequentialised.  

\begin{proof} 
  Consider the digraph $\H=(\Hdiff(x,y), \frus(x))$ and let
  $\delta:\Hdiff(x,y)\to \set{0,1,\ldots, m-1}$ be a topological ordering of the
  nodes of $\H$: $\forall j,i\in \Hdiff(x,y),\ (j,i)\in\frus(x) \Rar
  \delta(i)\leq \delta(j)$ s.t. if $j$ and $i$ do not belong to the same
  cycle in $\H$ (and thus do not belong to the same $x$-\fredC{} of $\G$), then
  $(j,i)\in\frus(x) \Rar \delta(i)< \delta(j)$.  Now, let
  $\Hdiff_t=\set{i\in\Hdiff(x,y),\ \delta(i)=t}$ and $x(0)=x$. Based on
  Lemma~\ref{lemma: adding frus}, an induction on $t<m$ proves that $\forall
  t<m,\ x(t)\trans x(t+1)=\bar{x(t)}{\Hdiff_t}$ is a transition of $\N$.
  \qed
\end{proof}

The next lemma considers the case where the only \frableC{s} of $\N$ are
Hamiltonian cycles.
\begin{lemma} Let $\N$ be a \Ban{} whose \frableC{s} all have node set
  $\V$. Then, either $\N$ has a unique transition $x\nonseq y$, or it has two
  $x\nonseq y$ and $y\nonseq x$. In the first case, every $i\in \nU(y)$ bears a
  positive loop $(i,i)\in\A$.  In both cases the endpoints of these transitions
  can be reached by no \as{} derivation.
\label{lemma: V}
\end{lemma}


\begin{proof}  Suppose that $x\nonseq y$ and $x'\nonseq y'$ are two \normal{}
  transitions. Using Proposition~\ref{prop:romp}, if $x'\neq y$, then
  $\W=\Hdiff(x,x')\subsetneq\U(x)=\V$ and $\Hdiff(x',y)=\V\setminus \W\subsetneq
  \U(x')=\V$. In this case, $x\trans x'\trans y$ is a derivation of $\N$
  involving smaller transitions than $x\trans y$, in contradiction with $x\trans
  y$ being \normal{}. Thus, if $x\nonseq y$ is not the unique \normal{}
  transition of $\N$, then the only other one is $y\nonseq x$.  For any
  \normal{} transition $z\nonseq z'=\bar{z}{\V}$, and $\forall i\in\V$, $z\trans
  \bar{z'}{i}$ is a transition of $\N$. By hypothesis and by
  Proposition~\ref{prop:romp}, it is \seqable{}: $z\seqRT \bar{z'}{i}$. Since
  $z\nonseq z'$ is not however, this implies $i\in\nU(\bar{z'}{i}),\, \forall
  i\in\V$. Thus, the endpoint of any \normal{} transition of $\N$ can be reached
  by no \as{} derivation. And since $\forall i\in \V,\ i\in\nU(\bar{y}{i})$, any
  $ i\in\nU(y)$ is such that $\sign(i,i)=+1$ by Lemma~\ref{lem:loop}.\qed
\end{proof}


\def\destroysLTF{ $\forall x\in\Bn,
\ \left \lbrace \begin{array}{lcl}
    \uh_0(x)&=&x_2\vee(x_0\wedge\neg x_1)\\\uh_1(x)&=&x_3\vee(\neg x_0\wedge
    x_1)\\\uh_2(x)&=&\neg x_0\wedge x_1\\ \uh_3(x) &=& x_0\wedge\neg
    x_1\end{array} \right.$}


\def\destroysSIG{ 
  \begin{tikzpicture}[trans/.style ={->,>=angle
        60,out=5,in=180},seq/.style={-open triangle 60,out=5,in=180}]
    \node  at (1,-0.2) (x4) {$0011$};   
    \node  at (3,1)   (x2a) {$0010$};   
    \node  at (3,-1)  (x2b) {$0001$};
    \node[ellipse,fill=black!10,]  at (5.9,-1.1) {\parbox{2.9cm}
      {\centerline{~}\vspace{0.2cm}{\footnotesize\centerline{\mbox{\textcolor{black!50}{\scc}}}}}};  
    \node  at (5.9,-1) (cfc)  {$\{x\in\B^4\,|\, x_0\vee  x_1=1\}$};
    \node[rectangle,fill=black!30]  at (7.5,0.5) {\parbox{2.8cm}
      {\centerline{~}\vspace{0.2cm}{\footnotesize\centerline{\mbox{\textcolor{black!60}{stable
                configuration}}}}}};   
    \node  at (7.5,0.7) (fp)  {$0000$}; 
    \path  (x4) edge[seq] (x2a) edge[seq] (x2b);
    \path  (x4) edge[seq,out=70,in=120]  (fp); 
    \path  (x4) edge[seq,out=-50,in=-170]   (5.2,-1.42); 
    \path (x2a) edge[seq,in=100] (cfc) edge[seq,in=160]  (fp); 
    \path (x2b) edge[seq,out=30,in=150] (cfc)
    edge[seq,out=30,in=-170] (fp); 
    %
    \node at   (7.8,0.1) (fpbis) {};            
    \path 
    (7.7,-1.1) edge[-,line width=1.5mm,out=10,in=-40] (fpbis) ; 
    \draw[-triangle 60,line width=1.5pt] (fpbis) -- +(.29,0); 
\end{tikzpicture}}


\def\destroysGT{ 
  \begin{tikzpicture}[seq/.style={-open triangle 60}]
    \node  at (0,-0.4)      (x11) {$0011$};   
    \node  at (-1.5,0.5) (x10) {$0010$};   
    \node  at (1.5,0.6)  (x01) {$0001$};
    \node[ellipse,fill=black!10,]  at (-4,-0.17) {\parbox{2.9cm}
      {\centerline{~}\vspace{0.2cm}{\footnotesize\centerline{\mbox{\textcolor{black!50}{\scc}}}}}};  
    \node  at (-4,0) (cfc)  {$\{x\in\B^4\,|\, x_0\vee  x_1=1\}$};
    \node[rectangle,fill=black!30]  at (3.5,-0.2) {\parbox{2.8cm}
      {\centerline{~}\vspace{0.2cm}{\footnotesize\centerline{\mbox{\textcolor{black!60}{stable
                configuration}}}}}};   
    \node  at (3.5,0) (fp)  {$0000$}; 
    \path  (x11) edge[seq,out=100,in=0]   (x10);
    \path  (x11) edge[seq,out=80,in=180]  (x01);
    \path  (x11) edge[seq,out=-10,in=-160] (2.3,-0.1);  
    \path  (x11) edge[seq,out=190,in=-20]  (-2.3,-0.1); 
    \path  (x10) edge[seq,out=180,in=20] (cfc) edge[seq,out=30,in=140]  (3.7,0.2); 
    \path  (x01) edge[seq,out=155,in=40] (cfc) edge[seq,out=0,in=140] (fp); 
    \node at   (2,-0.6) (fpbis) {};   
    \node at   (-3.1,-0.4) (cfcbis) {};            
    \path
    (cfcbis) edge[-,line width=1.5mm,out=-20,in=-160] (fpbis) ; 
    \draw[-triangle 60,line width=1.5pt] (2,-0.6) -- +(0.1,0.04); 
\end{tikzpicture}}



\def\freeLTF{$\forall x\in\Bn,\,\left \lbrace \begin{array}{lcl} \uf_0(x)&=& x_0\wedge
    \neg x_1\\ \uf_1(x)&=& \neg x_0\wedge x_1\end{array} \right. $}


\def\freeSIG{
  \begin{tikzpicture}[trans/.style ={->,>=angle
        60,out=5,in=180},seq/.style={-open triangle 60}]
    \node[rectangle,fill=black!30]  at (0,1.5){\parbox{1cm}{\centerline{~}\vspace{0.2cm}}};
    \node[rectangle,fill=black!30]  at (2,0){\parbox{1cm}{\centerline{~}\vspace{0.2cm}}};
    \node[rectangle,fill=black!30]  at (2,1.5){\parbox{1cm}{\centerline{~}\vspace{0.2cm}}};
    \node  at (0,0) (x11) {$11$};   
    \node  at (0,1.5) (x01) {$01$};   
    \node  at (2,0) (x10) {$10$};
    \node  at (2,1.5) (x00) {$00$};
    \path  (x11) edge[seq] (x01) ;
    \path  (x11) edge[seq] (x10); 
    \node at   (1.7,1.45) (x00bis) {};            
    \path (x11) edge[-,line width=1.5mm] (x00bis) ; 
    \draw[-triangle 60,line width=1.5pt] (x00bis) -- +(-.03,-0.025); 
\end{tikzpicture}}


\def\exfree{
  \begin{example}[\freeS{} to synchronism]
    \begin{tabular}{m{0.6\textwidth}m{0.6\textwidth}} 
      \freeLTF &  \scalebox{0.5}{\begin{picture}(0,0)%
\includegraphics{exfree.pdf}%
\end{picture}%
\setlength{\unitlength}{4144sp}%
\begingroup\makeatletter\ifx\SetFigFont\undefined%
\gdef\SetFigFont#1#2#3#4#5{%
  \reset@font\fontsize{#1}{#2pt}%
  \fontfamily{#3}\fontseries{#4}\fontshape{#5}%
  \selectfont}%
\fi\endgroup%
\begin{picture}(2295,1173)(5274,590)
\put(7291,974){\makebox(0,0)[lb]{\smash{{\SetFigFont{14}{16.8}{\rmdefault}{\mddefault}{\updefault}{\color[rgb]{0,0,0}$1$}%
}}}}
\put(7111,1604){\makebox(0,0)[lb]{\smash{{\SetFigFont{12}{14.4}{\rmdefault}{\mddefault}{\updefault}{\colorbox{white}{\color[rgb]{0,0,0}$+$}}%
}}}}
\put(6346,659){\makebox(0,0)[lb]{\smash{{\SetFigFont{12}{14.4}{\rmdefault}{\mddefault}{\updefault}{\colorbox{white}{\color[rgb]{0,0,0}$-$}}%
}}}}
\put(6301,1199){\makebox(0,0)[lb]{\smash{{\SetFigFont{12}{14.4}{\rmdefault}{\mddefault}{\updefault}{\colorbox{white}{\color[rgb]{0,0,0}$-$}}%
}}}}
\put(5491,974){\makebox(0,0)[lb]{\smash{{\SetFigFont{14}{16.8}{\rmdefault}{\mddefault}{\updefault}{\color[rgb]{0,0,0}$0$}%
}}}}
\put(5311,1604){\makebox(0,0)[lb]{\smash{{\SetFigFont{12}{14.4}{\rmdefault}{\mddefault}{\updefault}{\colorbox{white}{\color[rgb]{0,0,0}$+$}}%
}}}}
\end{picture}%
}
    \end{tabular}
    
    \begin{figure}
      \centerline{\freeSIG}
      \caption{}
      \label{SIG exfree}
    \end{figure}
    \label{Ex free}
  \end{example}
}



\def\growsLTF{$$\forall x\in\Bn,\,\left \lbrace \begin{array}{lcl} \uf_0(x)&=&
        \uh_0(x)\vee (x_4\wedge \neg x_1\wedge\neg x_3)\\ \uf_i(x)&=&
        \uh_i(x),\ \forall i\in \set{1,2,3}\\ \uf_4(x)&=& \neg(x_0\vee x_1\vee
        x_2\vee x_3)\wedge\neg x_4\end{array} \right. $$}


\def\growsSIG{
  \begin{tikzpicture}[trans/.style ={->,>=angle
        60,out=5,in=180},seq/.style={-open triangle
        60,out=-5,in=180,color=black!40},seqq/.style={-open triangle
        60,out=-5,in=180}]
    \node[color=black!40] at (0.5,-0.2) (x4) {$0011\mathbf{0}$};
    \node[color=black!40] at (2.5,-1.2) (x2a) {$0010\,\mathbf{0}$};
    \node[color=black!40] at (2.5,0.8) (x2b) {$0001\,\mathbf{0}$};
    \node[ellipse,fill=black!10]  at (5.9,0.5) {\parbox{3.2cm}
      {\centerline{~}\vspace{0.3cm}{\footnotesize\centerline{\mbox{\textcolor{black!50}{\scc}}}}}};  
    \node  at (5.9,0.5) (cfc)  {$\{x\in\B^5\,|\, x_0\vee  x_1=1\wedge x_4=\mathbf{0}\}$};
    \node[rectangle,fill=black!30]  at (7.5,-2.2) {\parbox{2.8cm}
      {\centerline{~}\vspace{2.8cm}{\footnotesize\centerline{\mbox{\textcolor{black!60}{~}}}}}};   
    \node  at (7.5,-0.8) (fp)  {$0000\,\mathbf{0}$}; 
    \path  (x4) edge[seq] (x2a) edge[seq] (x2b);
    \path  (x4) edge[seq,out=-65,in=-145]  (fp); 
    \path  (x4) edge[seq,out=60,in=130]   (5.2,0.8); 
    \path (x2a) edge[seq,in=-100] (5.4,0.15) edge[seq,in=-160]  (fp); 
    \path (x2b) edge[seq,out=-40,in=-150] (5.2,0.2)
    edge[seq,out=-40] (fp); 
    \node at   (8.4,-0.7) (fpbis) {};            
    \path (7.8,0.2) edge[-,line width=1.5mm,out=-10,in=30] (8.3,-0.75) ; 
    \draw[-triangle 60,line width=1.5pt] (fpbis) -- +(-.2,-0.2); 
    \node[color=black!40]  at (0.5,-4.2) (x4B) {$0011\mathbf{1}$};   
    \node[color=black!40]   at (2.5,-3)   (x2aB) {$0010\,\mathbf{1}$};   
    \node[color=black!40]  at (2.5,-5)  (x2bB) {$0001\,\mathbf{1}$};
    \node[ellipse,fill=black!10,]  at (5.9,-5)  {\parbox{3.2cm}
      {\centerline{~}\vspace{0.3cm}{\footnotesize\centerline{\mbox{\textcolor{black!50}{\scc}}}}}};  
    \node  at (5.9,-5) (cfcB)  {$\{x\in\B^5\,|\, x_0\vee  x_1=1\wedge x_4=\mathbf{1}\}$};
    \node  at (7.5,-3.5) (fpB)  {$0000\,\mathbf{1}$}; 
    \path  (x4B) edge[seq] (x2aB) edge[seq] (x2bB);
    \path  (x4B) edge[seq,out=67,in=150]  (fpB); 
    \path  (x4B) edge[seq,out=-50,in=-170]   (5.2,-5.42); 
    \path (x2aB) edge[seq,in=120] (cfcB) edge[seq,in=170]  (fpB); 
    \path (x2bB) edge[seq,out=30,in=150] (cfcB)
    edge[seq,out=30,in=-170] (fpB); 
    \node at   (8.1,-3.7) (fpbisB) {};            
    \path (7.5,-4.7) edge[-,line width=1.5mm,out=20,in=-50] (fpbisB) ; 
    \draw[-triangle 60,line width=1.5pt] (fpbisB) -- +(.3,-0.06); 
    \path  (fp) edge[seqq,out=-70,in=75] (fpB);
    \path  (fpB) edge[seqq,out=120,in=-110] (fp);
    \draw[seqq]   (fpB) to[,out=5,in=80] (8.9,-4.65)  to[out=-95,in=-40]  (7.2,-5.25); 
\end{tikzpicture}}


\def\growsGT{
  \begin{tikzpicture}[seq/.style={-open triangle 60,color=black!40},
  seqq/.style={-open triangle 60}]
    \node[color=black!40]  at (0,-0.4)   (x11) {$0011\mathbf{0}$};
    \node[color=black!40]  at (-1.5,0.5) (x10) {$0010\,\mathbf{0}$};
    \node[color=black!40]  at (1.5,0.5)  (x01) {$0001\,\mathbf{0}$};
    \node[ellipse,fill=black!10,]  at (-4,0) {\parbox{3.3cm}
      {\centerline{~}\vspace{0.4cm}{\footnotesize\centerline{\mbox{\textcolor{black!50}{\scc}}}}}};  
    \node  at (-4,0) (cfc)  {$\{x\in\B^5\,|\, x_0\vee  x_1=1\wedge x_4=\mathbf{0}\}$};
    \node[rectangle,fill=black!20]  at (3.5,-1.25) {\parbox{2.3cm}
      {\centerline{~}\vspace{2.3cm}{\footnotesize\centerline{\mbox{\textcolor{black!60}{~}}}}}};   
    \node  at (3.5,0) (fp)  {$0000\,\mathbf{0}$}; 
    \path  (x11) edge[seq,out=100,in=0]   (x10);
    \path  (x11) edge[seq,out=80,in=180]  (x01);
    \path  (x11) edge[seq,out=-10,in=180] (fp);  
    \path  (x11) edge[seq,out=190,in=-20]  (-2.3,-0.2); 
    \path  (x10) edge[seq,out=180,in=20] (cfc) edge[seq,out=30,in=140]  (3.7,0.2); 
    \path  (x01) edge[seq,out=155,in=40] (cfc) edge[seq,out=0,in=140] (fp); 
    \node at   (2.9,-0.3) (fpbis) {};   
    \node at   (-3.1,-0.4) (cfcbis) {};            
    \path (cfcbis) edge[-,line width=1.5mm,out=-20,in=-160] (fpbis) ; 
    \draw[-triangle 60,line width=1.5pt] (2.9,-0.3) -- +(0.1,0.04); 
    %
    %
    %
    %
    \node[color=black!40]  at (0,-2.9)   (x11B) {$0011\mathbf{1}$};
    \node[color=black!40]  at (-1.5,-2) (x10B) {$0010\,\mathbf{1}$};
    \node[color=black!40]  at (1.5,-2)  (x01B) {$0001\,\mathbf{1}$};
    \node[ellipse,fill=black!10,]  at (-4,-2.5) {\parbox{3.3cm}
      {\centerline{~}\vspace{0.4cm}{\footnotesize\centerline{\mbox{\textcolor{black!50}{\scc}}}}}};  
    \node  at (-4,-2.5) (cfcB)  {$\{x\in\B^5\,|\, x_0\vee  x_1=1\wedge x_4=\mathbf{1}\}$};
    \node  at (3.5,-2.5) (fpB)  {$0000\,\mathbf{1}$}; 
    \path  (x11B) edge[seq,out=100,in=0]   (x10B);
    \path  (x11B) edge[seq,out=80,in=180]  (x01B);
    \path  (x11B) edge[seq,out=-10,in=180] (fpB);  
    \path  (x11B) edge[seq,out=190,in=-20]  (-2.3,-2.7); 
    \path  (x10B) edge[seq,out=180,in=20] (cfcB) edge[seq,out=30,in=140]  (3.5,-2); 
    \path  (x01B) edge[seq,out=155,in=40] (cfcB) edge[seq,out=0,in=140] (fpB); 
    \node at   (2.9,-2.8) (fpbisB) {};   
    \node at   (-3.1,-2.9) (cfcbisB) {};            
    \path (cfcbisB) edge[-,line width=1.5mm,out=-20,in=-160] (fpbisB) ; 
    \draw[-triangle 60,line width=1.5pt] (2.9,-2.8) -- +(0.1,0.04); 
    \path  (fp)  edge[seqq,out=-70,in=75] (fpB);
    \path  (fpB) edge[seqq,out=120,in=-110] (fp);
    \draw[seqq]  (fpB) to[out=-130,in=0] (0,-3.8)  to[out=180,in=-40]
  (-3.5,-2.9); 
\end{tikzpicture}}


\def\exgrows{
  \begin{example}[\growsS{} to synchronism]
    \begin{tabular}{m{0.6\textwidth}m{0.6\textwidth}} 
      \growsLTF &  \scalebox{0.5}{INPUT FIGURE}
    \end{tabular}
    
    If $\uh_0(x)=1$, then $\uf_0(x)=1,\ \forall x\in\B^5$. Hence, if
    $\uh_0(x)\neq \uf_0(x)$, it must be that $\uh_0(x)=0=x_2=(x_0\wedge \neg
    x_1)$ and  $\uf_0(x)=1=(x_4\wedge \neg x_1\wedge\neg x_3)$. Thus,
    $\uh_0(x)\neq \uf_0(x)\ \implies\ x=00001$. 

    \begin{figure}
      \growsSIG
      \caption{
    }
      \label{SIG exgrows}
    \end{figure}    
    \label{Ex grows}
  \end{example}
}



\def\figcontrex{ \begin{figure} 
    \centering \hspace{1cm} \destroysGT \caption{Schematic representation of
    $\SIGbis$, the \sig{} of the \Ban{} of Example \ref{Ex:contrex} augmented
    with \normal{} transition $1100\nonseq 0000$ which updates both automata $0$
    and $1$ simultaneously. The shaded ellipse corresponds a strongly connected
    component which is terminal in $\SIG$ but not in $\SIGbis$ nor the \gig{}
    with the added possibility of $1100\nonseq 0000$.}  \label{SIG
    contrex} \end{figure}}

\def\inputexample{
\begin{example}[\destroys- and \grows-impact] 
  Let $\N=\set{f_i\,;\, i<4}$ be the \Ban{} of size $4$ whose \ltf{s} and signed
  structure are given below:

  {\begin{tabular}{m{0.56\textwidth}m{0.6\textwidth}} \destroysLTF
      & \scalebox{0.5}{\begin{picture}(0,0)%
\includegraphics{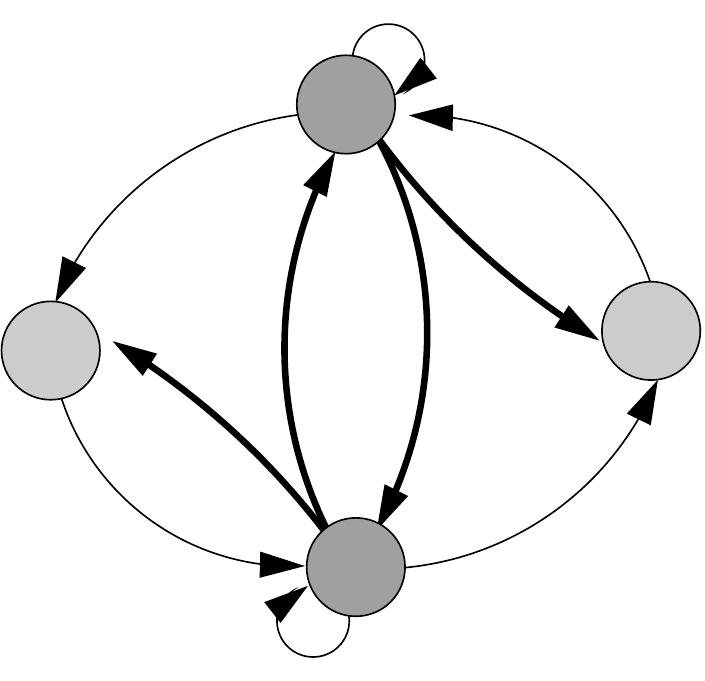}%
\end{picture}%
\setlength{\unitlength}{4144sp}%
\begingroup\makeatletter\ifx\SetFigFont\undefined%
\gdef\SetFigFont#1#2#3#4#5{%
  \reset@font\fontsize{#1}{#2pt}%
  \fontfamily{#3}\fontseries{#4}\fontshape{#5}%
  \selectfont}%
\fi\endgroup%
\begin{picture}(3211,3108)(6428,-1660)
\put(6616,-241){\makebox(0,0)[lb]{\smash{{\SetFigFont{14}{16.8}{\rmdefault}{\mddefault}{\updefault}{\color[rgb]{0,0,0}$2$}%
}}}}
\put(9361,-151){\makebox(0,0)[lb]{\smash{{\SetFigFont{14}{16.8}{\rmdefault}{\mddefault}{\updefault}{\color[rgb]{0,0,0}$3$}%
}}}}
\put(7966,884){\makebox(0,0)[lb]{\smash{{\SetFigFont{14}{16.8}{\rmdefault}{\mddefault}{\updefault}{\color[rgb]{0,0,0}$1$}%
}}}}
\put(8011,-1231){\makebox(0,0)[lb]{\smash{{\SetFigFont{14}{16.8}{\rmdefault}{\mddefault}{\updefault}{\color[rgb]{0,0,0}$0$}%
}}}}
\put(7111,659){\makebox(0,0)[lb]{\smash{{\SetFigFont{12}{14.4}{\rmdefault}{\mddefault}{\updefault}{\colorbox{white}{\color[rgb]{0,0,0}$+$}}%
}}}}
\put(7426,-556){\makebox(0,0)[lb]{\smash{{\SetFigFont{12}{14.4}{\rmdefault}{\mddefault}{\updefault}{\colorbox{white}{\color[rgb]{0,0,0}$-$}}%
}}}}
\put(6931,-916){\makebox(0,0)[lb]{\smash{{\SetFigFont{12}{14.4}{\rmdefault}{\mddefault}{\updefault}{\colorbox{white}{\color[rgb]{0,0,0}$+$}}%
}}}}
\put(7741,-1591){\makebox(0,0)[lb]{\smash{{\SetFigFont{12}{14.4}{\rmdefault}{\mddefault}{\updefault}{\colorbox{white}{\color[rgb]{0,0,0}$+$}}%
}}}}
\put(8956,614){\makebox(0,0)[lb]{\smash{{\SetFigFont{12}{14.4}{\rmdefault}{\mddefault}{\updefault}{\colorbox{white}{\color[rgb]{0,0,0}$+$}}%
}}}}
\put(8911,-916){\makebox(0,0)[lb]{\smash{{\SetFigFont{12}{14.4}{\rmdefault}{\mddefault}{\updefault}{\colorbox{white}{\color[rgb]{0,0,0}$+$}}%
}}}}
\put(8056,1289){\makebox(0,0)[lb]{\smash{{\SetFigFont{12}{14.4}{\rmdefault}{\mddefault}{\updefault}{\colorbox{white}{\color[rgb]{0,0,0}$+$}}%
}}}}
\put(7651,-196){\makebox(0,0)[lb]{\smash{{\SetFigFont{12}{14.4}{\rmdefault}{\mddefault}{\updefault}{\colorbox{white}{\color[rgb]{0,0,0}$-$}}%
}}}}
\put(8281,-196){\makebox(0,0)[lb]{\smash{{\SetFigFont{12}{14.4}{\rmdefault}{\mddefault}{\updefault}{\colorbox{white}{\color[rgb]{0,0,0}$-$}}%
}}}}
\put(8551,209){\makebox(0,0)[lb]{\smash{{\SetFigFont{12}{14.4}{\rmdefault}{\mddefault}{\updefault}{\colorbox{white}{\color[rgb]{0,0,0}$-$}}%
}}}}
\end{picture}%
} \end{tabular}} The \gig{} of
      this \Ban{} has two attractors: one unstable and one stable (configuration
      $0000$).  When $x_0=x_1=1$ and $x_2=x_3=0$, the simultaneous update of
      automata $0$ and $1$ has an effect that cannot be mimicked by a series of
      atomic updates (\cf \figref{SIG contrex}). If it could, strongly connected
      component $\atta{1100}$ would not be terminal in the \sig{}. From
      Propositions \ref{prop:critical cycles} and \ref{prop:romp} we know that
      this is essentially due to the positive cycle of length $2$ induced by
      automata $0$ and $1$.\smallskip
 
  Building on this example, we can derive an example of a \Ban{} with a
  \grows-impact transition. Indeed, consider a fifth automaton $i=4\in\V$ s.t.
  $ \uf_4(x)= \neg(x_0\vee x_1\vee x_2\vee x_3)\wedge\neg x_4$ and let
  $\uf_0(x)= \uh_0(x)\vee (x_4\wedge \neg x_1\wedge\neg x_3)$ and $\forall i\in
  \set{1,2,3}$, let $\uf_i(x)= \uh_i(x)$.  It can be checked that $\uh_0(x)\neq
  \uf_0(x)\ \implies\ x=00001$ and as a consequence, adding the same \normal{}
  transition as before to the \sig{} of this new \Ban{} yields the \tg{} below:

     \centerline{ \growsGT}
  \label{Ex:contrex}
\end{example}
}



\def\nonmonSIG{
  \begin{tikzpicture}[trans/.style ={->,>=angle
        60,out=5,in=180},seq/.style={-open triangle 60}]
    \node[rectangle,fill=black!30]  at (0,1.5){\parbox{1cm}{\centerline{~}\vspace{0.2cm}}};
    \node[rectangle,fill=black!30]  at (2,0){\parbox{1cm}{\centerline{~}\vspace{0.2cm}}};
    \node[rectangle,fill=black!30]  at (2,1.5){\parbox{1cm}{\centerline{~}\vspace{0.2cm}}};
    \node  at (0,0) (x11) {$11$};   
    \node  at (0,1.5) (x01) {$01$};   
    \node  at (2,0) (x10) {$10$};
    \node  at (2,1.5) (x00) {$00$};
    \path  (x11) edge[seq,out=105,in=-105] (x01) ;
    \path  (x01) edge[seq, out=-75,in=75] (x11) ;
    \path  (x11) edge[seq, out=15, in=165] (x10); 
    \path  (x10) edge[seq,out=-165,in=-15] (x11); 
    \node at   (1.7,1.95) (x00bis) {};            
    \path (x11) edge[-,line width=1.5mm] (x00bis) ; 
    \draw[-triangle 60,line width=1.5pt] (x00bis) -- +(-.03,-0.03); 
\end{tikzpicture}}


\def\nonmonSIGgen{
  \begin{tikzpicture}[trans/.style ={->,>=angle
        60,out=5,in=180},seq/.style={-open triangle 60}]
    \node[rectangle,fill=black!30]  at (0,1.5){\parbox{1.2cm}{\centerline{~}\vspace{0.2cm}}};
    \node[rectangle,fill=black!30]  at (2.5,0){\parbox{1.2cm}{\centerline{~}\vspace{0.2cm}}};
    \node[rectangle,fill=black!30]  at (2.5,1.5){\parbox{1.2cm}{\centerline{~}\vspace{0.2cm}}};
    \node  at (0,0) (x11) {$x=11$};   
    \node  at (0,1.5) (x01) {$\bar{x}{0}=01$};   
    \node  at (2.5,0) (x10) {$\bar{x}{1}=10$};
    \node  at (2.5,1.5) (x00) {$\bar{x}{\V}=00$};
    \path  (x11) edge[seq,out=105,in=-105] (x01) ;
    \path  (x01) edge[seq, out=-75,in=75] (x11) ;
    \path  (x11) edge[seq, out=15, in=165] (x10); 
    \path  (x10) edge[seq,out=-165,in=-15] (x11); 
    \node at   (2.1,1.3) (x00bis) {};            
    \path (x11) edge[-,line width=1.5mm] (x00bis) ; 
    \draw[-triangle 60,line width=1.5pt] (x00bis) -- +(-.03,-0.02); 
\end{tikzpicture}}

\def\nonmonSIGGEN{
  \begin{tikzpicture}[trans/.style ={->,>=angle
        60,out=5,in=180},seq/.style={-open triangle 60}]
    \node[rectangle,fill=black!30]  at (0,1.5){\parbox{1.6cm}{\centerline{~}\vspace{0.2cm}}};
    \node[rectangle,fill=black!30]  at (2.5,0){\parbox{1.6cm}{\centerline{~}\vspace{0.2cm}}};
    \node[rectangle,fill=black!30]  at (2.5,1.5){\parbox{1.6cm}{\centerline{~}\vspace{0.2cm}}};
    \node  at (0,0) (x11) {$x$};   
    \node  at (0,1.5) (x01) {$\bar{x}{0}=\bar{y}{1}$};   
    \node  at (2.5,0) (x10) {$\bar{x}{1}=\bar{y}{0}$};
    \node  at (2.5,1.5) (x00) {$\bar{x}{\set{0,1}}=y$};
    \path  (x11) edge[seq,out=105,in=-105] (x01) ;
    \path  (x01) edge[seq, out=-75,in=75] (x11) ;
    \path  (x11) edge[seq, out=15, in=165] (x10); 
    \path  (x10) edge[seq,out=-165,in=-15] (x11); 
    \node at   (2.1,1.3) (x00bis) {};            
    \path (x11) edge[-,line width=1.5mm] (x00bis) ; 
    \draw[-triangle 60,line width=1.5pt] (x00bis) -- +(-.03,-0.02); 
\end{tikzpicture}}

\def\nonmon{
\begin{wrapfigure}[16]{r}{4.5cm}
\centering \nonmonSIGGEN
\caption{\sig{} of a strongly connected \Ban{} $\N$ of size $2$ s.t. $\uf_0,\uf_1\in\set{x\mapsto x_0\oplus
        x_1,\, x\mapsto \neg(x_0\oplus x_1)}$, augmented with normal transition
        $x\nonseq y$. $\N$ is \destroys{-}sensitive to synchronism.}
\label{fig:nonmon}   
\end{wrapfigure}}


\def\monSIGGEN{
  \begin{tikzpicture}[trans/.style ={->,>=angle
        60,out=5,in=180},seq/.style={-open triangle 60},seqS/.style={open triangle 60-open triangle 60}]
     \node[rectangle,fill=black!30]  at (8.5,0){\parbox{1.2cm}{\centerline{~}\vspace{0.2cm}}};
    \node  at (0.5,0)  (y2)  {$ab\neg c$};
    \node  at (2,-1) (x12) {$\neg ab\neg c$};
    \node  at (2,1)  (x02) {$a\neg b\neg c$};
    \node  at (4,1)  (x0)  {$a\neg bc$};
    \node  at (4,-1) (x1)  {$\neg abc$};
    \node  at (6,0)  (x)   {$x=\neg a\neg b c$};
    \node  at (3.3,0)  (x2)  {$\neg a\neg b \neg c$};
    \node  at (8.5,0) (y)   {$y=abc$};
    \path  (x) edge[seq] (x0) edge[seq] (x1)  ;
    \path  (x0) edge[seqS] (x02)  ;
    \path  (x1) edge[seqS] (x12)  ;
    \path  (y2) edge[seq] (x12) edge[seq] (x02)  ;
    \path  (x02) edge[seq] (x2)  ;
    \path  (x12) edge[seq] (x2)  ;
    \path  (x2) edge[seqS] (x)  ;
    \draw[seq](y2) to[,out=70,in=180] (2,1.5)  to[out=0,in=180] (7,1.5)
        to[out=0,in=110] (y); 
%
    \path (x) edge[-,line width=1.5mm] (7.7,0) ; 
    \draw[-triangle 60,line width=1.5pt] (7.7,0) -- +(.1,0); 
\end{tikzpicture}}

\def\exmon{ 
\begin{figure}
  \begin{tabular}{cc}
  \scalebox{0.5}{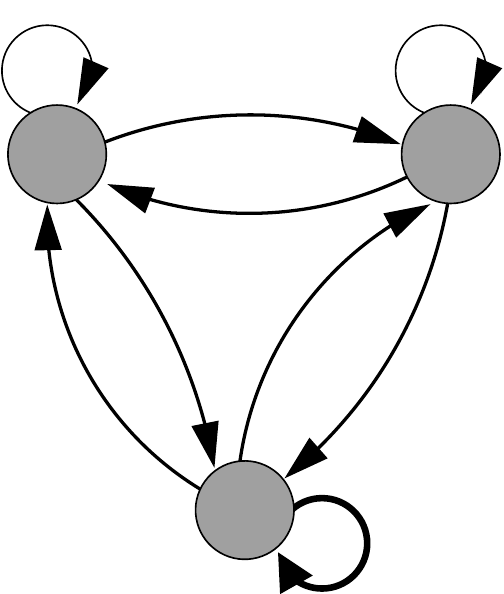} ~~&\monSIGGEN
  \end{tabular}
\caption{Generic signed structure ($\forall
  i,j\in\V,\, \sbool_{ji}=\sign(j,i)=\sign(i,j)$) and modified \sig{} $\SIGbis$
  of all monotone \Ban{s} of size $3$ that are very sensitive
  (necessarily \destroys-sensitive) to synchronism, \eg $\N=\set{\uf_0:x\mapsto
  x_2\vee(x_0\wedge \neg x_1),\, \uf_1:x\mapsto x_2\vee(\neg x_0\wedge
  x_1),\, \uf_2:x\mapsto \neg x_2\wedge (x_0\vee x_1)}$. For all instances of
  these \Ban{s}, in the starting point $x$ of the normal transition,
  $\uf_2\big(\,\uf_0(x)\,\uf_1(x)\,x_2\,\big)\in\set{x_0\oplus
  x_1,\, \neg(x_0\oplus x_1)}$.}
\label{fig:new}
\end{figure}
}

\subsection*{Impact of  synchronous transitions}

Let us introduce some new vocabulary to describe the \tg{s} of $\N$. Stable
configurations and terminal strongly connected components of these graphs are
called \DEF{attractors} (abusing language because it may be that an attractor
doesn't attract anything).  Attractors that are not stable configurations are
said to be unstable.  Now, while presenting notations relative to the \sig{}
(under-scripted by '$\mathrm{a}$'), let us continue introducing terminology
relative to any \tg{} of $\N$, in particular its \gig{}. $\forall x\in\Bn$, we
let $\orbita(x)=\set{y\in\Bn\,;\, x\seqRT y}$ and $\Ba(x)=\set{y\,;\, y\seqRT
  x}$. Also, we let $\Aa(x)$ denote the set of attractors that $x$ can reach in
$\SIG$. We say that a configuration $x$ is \DEF{recurrent} when it belongs to an
attractor and we denote this attractor by $\atta{x}$ (then,
$\Aa(x)=\set{\atta{x}}$). The basin of an attractor $\atta{x}$ is
$\Ba(\atta{x})=\Ba(x)\setminus \atta{x}$. Non-recurrent configurations are
called \DEF{transient}.
 \smallskip


Let us consider an arbitrary synchronous transition $x\pll y$ of $\N$ and let
$\SIGbis=(\Bn, \AT\cup\set{(x,y)})$ denote the transition graph obtained by
adding this transition to the \sig{} $\SIG$. We introduce notations $\Ae(z)$,
$\Be(z)$, $\orbit(z)$ and $\att{x}$ relative to $\SIGbis$ naturally as we did
above for $\SIG$.  In the sequel, we say that an attractor $A$ of $\SIG$ is
\emph{destroyed} by $x\pll y$ if all its configurations are transient in
$\SIGbis$. Generally, since $\forall z\in\Ba(x)\cup\set{x},\ \Ae(z)=\Aa(z)\cup
\Aa(y)\comma$ the addition of $x\pll y$ to $\SIG$ can have several possible
consequences on the asymptotic evolution of $\N$ starting in a configuration
$z\in\Ba(x)\cup\set{x}$. We list them now exhaustively.

\begin{itemize}
\item[1.] We say that it has \DEF{no impact} when $x$ is transient in $\SIG$ and
  $\Aa(y)\subset \Aa(x)=\Ae(x)$. In this case, $x\pll y$ `only' adds to $\SIG$
  some new derivations from $x$ to the configurations of the orbit $\orbita(x)$
  of $x$.  It does not change the result of any network evolution. In
  particular, if $x\pll y$ is \seqable{}, then it shortcuts some derivations
  starting in $x$. But on the contrary, it can also deviate some derivations
  (when $\exists z\in\orbita(x)\cap\orbita(y)$ s.t.  $y \seqRT z$ is no shorter
  than $x\seqRT z$).
\end{itemize}

Obviously, all \sy{} transitions $x\pll y$ that \emph{do} have an impact on the
asymptotic evolution of $\N$ are \normal.


\begin{itemize}
\item[2.] We say that transition $x\nonseq y$ has little or \DEF{\free-impact}
  (\cf \figref{ex:free}) if $x$ is transient in $\SIG$ and $\Ae(x)=\Aa(x)\cup
  \Aa(y)\neq\Aa(x)$. Here, the addition of $x\nonseq y$ adds some new degrees of
  \underline{{\bf f}}reedom to the asymptotic outcomes of the evolutions of
  $\N$ from any configuration $z\in\Ba(x)\cup\set{x}$.Thus, it causes the growth
  of the basins $\Ba(A),\ A\in \Aa(y)$.
\end{itemize}
\begin{figure}
  \begin{tabular}{m{0.4\textwidth}c@{\hspace{0.5cm}}c}
    \freeLTF &  \raisebox{-0.5cm}{\scalebox{0.5}{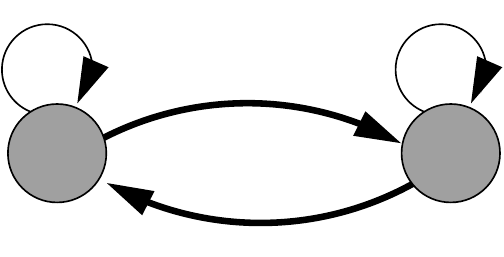}} & \raisebox{-1cm}{\freeSIG}
  \end{tabular}
  \caption{Transition functions, structure and modified \sig{} $\SIGbis$ of a
    \Ban{} $\N$ whose \normal{} transition $11\nonseq 00$ has \free-impact. This
    is a special case of \free-impact induced by one critical, Hamiltonian cycle
    with no automata outside of it. Its impact, precisely, consists in making
    reachable an reachable attractor (\cf Lemmas \ref{lemma: V} and \ref{lemma: V sensitive}).}
\label{ex:free}
\end{figure}
Note that with addition of synchronous transitions that have no or \free-impact,
the set of recurrent configurations of $\SIG$ equals that of $\SIGbis$.



\begin{itemize}
\item[3.] We say that transition $x\nonseq y$ has \DEF{\grows-impact} (\cf the
  end of Example \ref{Ex:contrex}) on the asymptotic evolution of $\N$ when,
  in $\SIG$, $x$ is recurrent, $y$ is transient and $\Aa(y)=
  \Aa(x)=\{\atta{x}\}$. In this case, $y$ is connected to $x$ and recurrent in
  $\SIGbis$.  The addition of $x\nonseq y$ to $\SIG$ causes attractor $\atta{x}$
  to absorb all derivations from $y$ to $\atta{x}$ and \underline{{\bf g}}row
  into $\att{x}=\att{y}$ without being destroyed.

\item[4.] We say that transition $x\nonseq y$ has \DEF{\destroys-impact} (\cf
  Example \ref{Ex:contrex}) if $x$ and $y$ are both recurrent in $\SIG$ and
  $\Aa(y)\setminus \Aa(x)\neq \emptyset$. In this case, the addition of
  $x\nonseq y$ \underline{{\bf d}}estroys the unstable attractor $\atta{x}$ by
  emptying it into (the basins of) the attractors $A\in \Aa(y)\setminus \Aa(x)$.

  %
\end{itemize}

\figcontrex

It can be checked that the four types of impact listed above are disjoint and
cover all possible cases. It follows as a particular consequence that a unique
\normal{} transition is not enough to merge attractors.  Let us emphasise that a
configuration that is recurrent in the \sig{} can become transient with the
addition of synchronism (if $\N$ has a \normal{} transition with
\destroys-impact), in particular, it can become transient in the
\gig{}. Conversely, synchronism can turn a transient configuration into a
recurrent one (if $\N$ has a transition with \grows-impact). Synchronism can
however not create new attractors from scratch. Indeed, if all configurations of
a set $X\subset\Bn$ are transient in the \sig{}, then in the \gig{} as well as
in its sub-graph the \sig{}, there necessarily exists a derivation outgoing $X$.
\smallskip

The addition of $x\nonseq y$ to the \sig{} has no or little (\ie \free-) impact
when $x$ is transient in the \sig{}. To change the asymptotics of $\N$ (rather
than just some of its evolutions towards it), $x\nonseq y$ needs to have \grows-
or \destroys-impact. And for this, in the \sig{}, an unstable attractor
$\atta{x}$ is needed. It can only be induced by a negative cycle in the
structure $\G$ of $\N$ \cite{Richard2010}.  Further, considering Hamiltonian
critical cycles again as in Lemma~\ref{lemma: V}, the last point of this section
evidences the need to embed \frableC{s} in a larger, structural 'environment' to
obtain \grows- and \destroys-impact transitions.  In other terms, $\N$ must have
a \frableC{} $\C=(\V_\C,\A_\C)$ as well as nodes $i\in\V\setminus\V_\C\neq
\emptyset$ outside of it if the addition of synchronism is to significantly
impact on its behaviour and change its asymptotics.

\begin{lemma} Let $\N$ be a \Ban{} with no  \normal{} transitions of  size
  smaller than its size $n$. Then, any transition $x\nonseq y$ either has no
  impact on the asymptotics of $\N$ or it has \free-impact. In the latter case,
  $y$ is a stable with an empty basin $\Ba(y)=\emptyset$ in the \sig{} and all
  nodes of the structure $\G$ of $\N$ have a positive loop.
  \label{lemma: V sensitive}
\end{lemma}

\begin{proof} Let $x\nonseq y=\bar{x}{\V}$ be a \normal{} transition of
  $\N$. Since $\forall i\in\V,\ \V\setminus \set{i}\subset \U(x)=\V$,
  Proposition~\ref{prop:romp} implies $x\seqRT \bar{y}{i},\ \forall
  i\in\V$. Thus, $\forall z\in\Bn,\ y\seq z \Rar\ x\seqRT z$. And either
  $\U(y)\neq \emptyset$ in which case $\Aa(y)\subset\Aa(x)$ (and $x\nonseq y$
  has no impact), or $y$ is stable in which case,  by Lemma~\ref{lemma: V}, its basin is
  empty  and all automata  bear a positive loop.\qed
\end{proof}

\begin{example}[\destroys- and \grows-impact] 
  Let $\N=\set{f_i\,;\, i<4}$ be the \Ban{} of size $4$ whose \ltf{s} and signed
  structure are given below:

  {\begin{tabular}{m{0.56\textwidth}m{0.6\textwidth}} \destroysLTF
      & \scalebox{0.5}{\begin{picture}(0,0)%
\includegraphics{ContrexPos2.pdf}%
\end{picture}%
\setlength{\unitlength}{4144sp}%
\begingroup\makeatletter\ifx\SetFigFont\undefined%
\gdef\SetFigFont#1#2#3#4#5{%
  \reset@font\fontsize{#1}{#2pt}%
  \fontfamily{#3}\fontseries{#4}\fontshape{#5}%
  \selectfont}%
\fi\endgroup%
\begin{picture}(3211,3108)(6428,-1660)
\put(6616,-241){\makebox(0,0)[lb]{\smash{{\SetFigFont{14}{16.8}{\rmdefault}{\mddefault}{\updefault}{\color[rgb]{0,0,0}$2$}%
}}}}
\put(9361,-151){\makebox(0,0)[lb]{\smash{{\SetFigFont{14}{16.8}{\rmdefault}{\mddefault}{\updefault}{\color[rgb]{0,0,0}$3$}%
}}}}
\put(7966,884){\makebox(0,0)[lb]{\smash{{\SetFigFont{14}{16.8}{\rmdefault}{\mddefault}{\updefault}{\color[rgb]{0,0,0}$1$}%
}}}}
\put(8011,-1231){\makebox(0,0)[lb]{\smash{{\SetFigFont{14}{16.8}{\rmdefault}{\mddefault}{\updefault}{\color[rgb]{0,0,0}$0$}%
}}}}
\put(7111,659){\makebox(0,0)[lb]{\smash{{\SetFigFont{12}{14.4}{\rmdefault}{\mddefault}{\updefault}{\colorbox{white}{\color[rgb]{0,0,0}$+$}}%
}}}}
\put(7426,-556){\makebox(0,0)[lb]{\smash{{\SetFigFont{12}{14.4}{\rmdefault}{\mddefault}{\updefault}{\colorbox{white}{\color[rgb]{0,0,0}$-$}}%
}}}}
\put(6931,-916){\makebox(0,0)[lb]{\smash{{\SetFigFont{12}{14.4}{\rmdefault}{\mddefault}{\updefault}{\colorbox{white}{\color[rgb]{0,0,0}$+$}}%
}}}}
\put(7741,-1591){\makebox(0,0)[lb]{\smash{{\SetFigFont{12}{14.4}{\rmdefault}{\mddefault}{\updefault}{\colorbox{white}{\color[rgb]{0,0,0}$+$}}%
}}}}
\put(8956,614){\makebox(0,0)[lb]{\smash{{\SetFigFont{12}{14.4}{\rmdefault}{\mddefault}{\updefault}{\colorbox{white}{\color[rgb]{0,0,0}$+$}}%
}}}}
\put(8911,-916){\makebox(0,0)[lb]{\smash{{\SetFigFont{12}{14.4}{\rmdefault}{\mddefault}{\updefault}{\colorbox{white}{\color[rgb]{0,0,0}$+$}}%
}}}}
\put(8056,1289){\makebox(0,0)[lb]{\smash{{\SetFigFont{12}{14.4}{\rmdefault}{\mddefault}{\updefault}{\colorbox{white}{\color[rgb]{0,0,0}$+$}}%
}}}}
\put(7651,-196){\makebox(0,0)[lb]{\smash{{\SetFigFont{12}{14.4}{\rmdefault}{\mddefault}{\updefault}{\colorbox{white}{\color[rgb]{0,0,0}$-$}}%
}}}}
\put(8281,-196){\makebox(0,0)[lb]{\smash{{\SetFigFont{12}{14.4}{\rmdefault}{\mddefault}{\updefault}{\colorbox{white}{\color[rgb]{0,0,0}$-$}}%
}}}}
\put(8551,209){\makebox(0,0)[lb]{\smash{{\SetFigFont{12}{14.4}{\rmdefault}{\mddefault}{\updefault}{\colorbox{white}{\color[rgb]{0,0,0}$-$}}%
}}}}
\end{picture}%
} \end{tabular}} The \gig{} of
      this \Ban{} has two attractors: one unstable and one stable (configuration
      $0000$).  When $x_0=x_1=1$ and $x_2=x_3=0$, the simultaneous update of
      automata $0$ and $1$ has an effect that cannot be mimicked by a series of
      atomic updates (\cf \figref{SIG contrex}). If it could, strongly connected
      component $\atta{1100}$ would not be terminal in the \sig{}. From
      Propositions \ref{prop:critical cycles} and \ref{prop:romp} we know that
      this is essentially due to the positive cycle of length $2$ induced by
      automata $0$ and $1$.\smallskip
 
  Building on this example, we can derive an example of a \Ban{} with a
  \grows-impact transition. Indeed, consider a fifth automaton $i=4\in\V$ s.t.
  $ \uf_4(x)= \neg(x_0\vee x_1\vee x_2\vee x_3)\wedge\neg x_4$ and let
  $\uf_0(x)= \uh_0(x)\vee (x_4\wedge \neg x_1\wedge\neg x_3)$ and $\forall i\in
  \set{1,2,3}$, let $\uf_i(x)= \uh_i(x)$.  It can be checked that $\uh_0(x)\neq
  \uf_0(x)\ \implies\ x=00001$ and as a consequence, adding the same \normal{}
  transition as before to the \sig{} of this new \Ban{} yields the \tg{} below:

     \centerline{ \growsGT}
  \label{Ex:contrex}
\end{example}

\subsection*{Synchronism sensitivity}

On the basis of the previous classification of the impact of \sy{} transitions,
we now take a more abstract point of view to propose a list of the different
types of sensitivity that a \Ban{} may have to (the addition of) synchronism.
Naturally, we say that $\N$ has \DEF{no sensitivity to synchronism} if none of
its \normal{} synchronous transitions has any impact (\cf Point 1 in the
previous section). We say that it has \DEF{\free{-} and \growsS{} to
  synchronism} when, respectively, it has \normal{} transitions with \free{-}
and \grows{-}impact (\cf Point 2 and 3). When $\N$ has \normal{} transitions
with \destroys-impact (\cf Point 4), two cases may occur. Indeed, let $x\nonseq
y$ be a \normal{} transition of $\N$ that has \destroys-impact. Then, there may
be another \destroys-impact \normal{} transition $y'\nonseq x'$ such that
$\atta{x}=\atta{x'}\neq \atta{y}=\atta{y'}$, \ie $x$ and $x'$ on the one hand,
and $y$ and $y'$ on the other belong to the same unstable attractors. In this
case, the two \normal{} transitions $x\nonseq y$ and $y'\nonseq x'$ cause 
attractors $\atta{x}$ and $ \atta{y}$ of the \sig{} to \underline{{\bf m}}erge
($\att{x}=\att{y}=\atta{x}\uplus\atta{y}$). Hence, $\N$ is said to have
\DEF{\mergesS{} to synchronism}. If there is no other \normal{} transition
connecting $\atta{y}$ to $\atta{x}$, then attractor $\atta{x}$ is effectively
\underline{{\bf d}}estroyed by the addition of $x\nonseq y$ and $\N$ is said to
have \DEF{\destroysS{} to synchronism}.  This and the results presented above as
well as, notably, the series of remarks made at the end of the previous section
yield Proposition \ref{prop:sensitivity} below.

\begin{proposition}
  \begin{enumerate}
  \item Sensitivity to synchronism requires the existence of a \frableC{}, and
    thus of an positive cycle with an even length or a negative cycle with an
    odd length.\label{need frableC} 

  \item \grows{-}, \destroys{-} and \mergesS{} require the existence of a
    \frableC{} of length strictly smaller than the \Ban{} size as well as of a
    negative cycle.\label{need more}

  \item Unless $\N$ has a Hamiltonian \frableC{} and positive loops on all of
    its automata, to have \freeS{}, $\N$ also needs to have a \frableC{} of
    length strictly smaller than the \Ban{} size.\label{need F} 
        
  \end{enumerate}
\label{prop:sensitivity}
\end{proposition}

\subsection*{Sensitivity to synchronism \& non-monotony}\label{sec:LM}

Obviously, to be sensitive to synchronism, a \Ban{} must involve at least two
automata.  It can be checked that there are no monotone \Ban{s} of size $2$ that
are \destroys- or \merges-sensitive (we say \DEF{very sensitive}) to
synchronism, but there are some non-monotone ones (\cf
\figref{fig:nonmon}).\smallskip


\nonmon

 However, interestingly, the monotone,
\destroys-sensitive \Ban{} of Example \ref{Ex:contrex} actually also involves
non-monotone actions. Indeed, it only involves a few monotone individual
interactions between four automata but these are architectured into a
\emph{widget} that can globally \emph{mimic} a punctual non-monotone action in
the right configuration and with the right synchronous update of automata
states.  More precisely but informally, in this widget, a non-monotone action is
\emph{structurally} split into two parts. These two parts consist in the two
halves of a {\sc xor}: $(x_0\,x_1)\mapsto x_0\wedge \neg x_1$ and
$(x_0\,x_1)\mapsto \neg x_0\wedge x_1$. They are encoded separately in the
\ltf{s} $\uf_0$ and $\uf_1$ of two different automata connected by what can be a
\frableC{} by Proposition \ref{prop:critical cycles}.  When the controls on
these two parts are lifted (\ie when $x_2=x_3=0$ so that we do indeed have
$\uf_0(x)=x_0\wedge \neg x_1$ and $\uf_1(x)=\neg x_0\wedge x_1$), the
synchronous update of automata $0$ and $1$ simultaneously applies $\uf_0$ and
$\uf_1$. Instantly, this amounts to combining influences underwent by $0$ and
$1$ by ``simulating'' a {\sc or} connector between their \ltf{s}, thereby
outputting the global action $\uf_0(x)\vee \uf_1(x)$. Precisely, this puts
together the two halves of a {\sc xor} with a $\vee$ and produces a global
non-monotone action.
Examining the widget of Example \ref{Ex:contrex}, one can notice that the automata
that it involves have different roles.  Roughly, automata $0$ and $1$ encode the
non-monotone action mentioned above.  The role of automata $2$ and $3$ is to
make ``use'' of it and ensure the necessary unstable attractor. This
attractor is made dependent on automata $0$ and $1$ by requiring $x_0\vee
x_1=1$. More precisely, the widget is designed so that the
unstable attractor is characterised by this condition. In the \sig{}, if the
condition becomes true, it remains true. Every configuration $x$ such that
$x_0=x_1=0$ reaches the stable configuration. \medskip 

\begin{figure}
  \begin{tabular}{cc}
  \scalebox{0.5}{\input{exmon.pdf_t}} ~~&\monSIGGEN
  \end{tabular}
\caption{Generic signed structure ($\forall
  i,j\in\V,\, \sbool_{ji}=\sign(j,i)=\sign(i,j)$) and modified \sig{} $\SIGbis$
  of all monotone \Ban{s} of size $3$ that are very sensitive
  (necessarily \destroys-sensitive) to synchronism, \eg $\N=\set{\uf_0:x\mapsto
  x_2\vee(x_0\wedge \neg x_1),\, \uf_1:x\mapsto x_2\vee(\neg x_0\wedge
  x_1),\, \uf_2:x\mapsto \neg x_2\wedge (x_0\vee x_1)}$. For all instances of
  these \Ban{s}, in the starting point $x$ of the normal transition,
  $\uf_2\big(\,\uf_0(x)\,\uf_1(x)\,x_2\,\big)\in\set{x_0\oplus
  x_1,\, \neg(x_0\oplus x_1)}$.}
\label{fig:new}
\end{figure}

These remarks suggest that there is a tight relationship between significant
sensitivity to synchronism and non-monotony\footnote{Notably, the example given
  by Robert in \cite{Robert1995B} to illustrate the ``bursting of attractors''
  caused by parallelisation (which agrees with \destroysS{}) in a deterministic
  setting, can also be shown to involve non-monotony (it has size $3$).}. Let us
add that the smallest monotone \Ban{s} that are sensitive to synchronism have
size $3$. They are monotone encodings of the non-monotone sensitive \Ban{s} of
size $2$ (\cf \figref{fig:nonmon}). This is proven by building around a normal
transition $x\in\B^3\nonseq y\in\B^3$ (that must satisfy $\Hdist(x,y)=2<3$ by
Lemma~\ref{lemma: V sensitive}) and substantially exploiting Lemmas
\ref{lem:loop} and \ref{lemma: adding frus} that hold with the hypothesis of
monotony. From this we derive in particular that all such \Ban{s} have an \sig{}
and a signed structure of the form of those represented in \figref{fig:new}, and
they have \destroysS{}.




\subsection*{Conclusion and perspectives}

Intuitively, for monotone \Ban{s}, frustrations and thus local instabilities are
best maintained by synchronism.  In particular, the parallel update schedule
systematically exploits all this momenta in each network configuration to
perform all possible changes. Thus, it is known to have a tendency to induce
what are sometimes considered as \emph{artefact behaviours}, unstable and
dependant on its strong constraint of synchronism.  In \cite{LATA2012} we
conjectured and argued that the more ``intricate'' the network structure, that
is, the more interconnected are its underlying cycles, the less chances do local
instabilities have to be sustained. And this was observed under the parallel
update schedule which is especially good at maintaining instabilities, so it
seems that cycle intersections have a strong propensity to reduce the
sensitivity of network behaviours to one of the characteristic effects of
synchronism. In this paper, we investigated further in these lines, focusing on
synchronism and its effect at a more minute level, that of elementary
transitions. And we also related synchronism sensitivity to structural
properties, namely the existence of critical cycles of positive sign and even
length or of negative sign and odd length embedded in a particular
environment. Also, notably, we have provided an example to evidence that
synchronism in itself may indeed impact significantly on the asymptotic
behaviour of a network: not only can it modify transient behaviours and make
attractors grow, it can also destroy unstable attractors. Contrary to some
traditional intuitions, asymptotically, asynchronism does not necessarily
guaranty a minimum of local instability. Synchronism too can filter
instabilities in asymptotic behaviours.  The disregard that synchronism has had
in theoretical modelling fields supports the importance of this: the time flow
mentioned in the introduction~--~accounted for by the way automata states are
updated, in particular synchronously or not~--~is a determining parameter of
networks behaviours.  What is more, we have argued in the previous section that
non-monotony (typically and minimally captured by a logical {\sc xor} connector)
might be necessary to manifest significant sensitivity to synchronism.  It seems
that monotone influences can be decomposed and stretched out in time whereas
non-monotone influences require specific inputs at a precise instant (otherwise
they act as a monotone influences, so these are defining conditions). Thus,
non-monotony relies on dynamics. And it can be hardwired in the network
definition or it may be \emph{mimicked} punctually by synchronism.  In the
latter case, 'time flow' assembles basic, hardwired operators in a way that is
otherwise impossible. \medskip

We have proven that sensitivity to synchronism requires strong structural
properties. It deserves to be highlighted that this can also be interpreted in
favour of asynchronous formalisations which advantageously yield less voluminous
behaviour descriptions. Indeed, \emph{most of the time}, synchronism only
shortcuts \as{} trajectories. However again, focusing on the asymptotics of a
network behaviour, if complexity is more of an issue than exhaustivity, then
deterministic update schedules such as the parallel update schedule also deserve to
be investigated and studied for themselves.  Thus, this work calls for further
studies to better gauge the determining capacity of update schedules on the
behaviours of networks. We have put forward (especially with Example
\ref{Ex:contrex}) the existence of a ``criticality'' that involves updates whose
effect is to decrease suddenly and non-reversibly the number of local
instabilities. This establishes a new point of view on how networks
work and this way, raises many new questions, \eg \emph{how must interactions
  between unstable automata be organised if their updates are to be
  consequential?} and \emph{how, generally, do local, punctual instabilities
  relate to the global (asymptotic) instability of a network?}.

 \end{document}